\documentclass[11pt]{article}
\usepackage{amsmath,amsfonts, amssymb,amsthm}

\title{Optimal Pricing is Hard}

\author {
Constantinos Daskalakis\thanks{Supported by a Sloan Foundation Fellowship, a Microsoft Research Faculty Fellowship, and NSF Award CCF-0953960 (CAREER) and CCF-1101491.}\\
EECS, MIT \\
\tt{costis@mit.edu}
\and
Alan Deckelbaum\thanks{Supported by Fannie and John Hertz Foundation Daniel Stroock Fellowship and NSF Award CCF-1101491.}\\
Math, MIT\\
\tt{deckel@mit.edu}
\and
Christos Tzamos\thanks{Supported by NSF Award CCF- 1101491.}\\
EECS, MIT\\
\tt{tzamos@mit.edu}
}
\newtheorem{theorem}{Theorem}

\newtheorem{definition}{Definition}

\newtheorem{remark}{Remark}

\begin{document}
\maketitle

\begin{abstract}
We show that computing the revenue-optimal deterministic auction in
unit-demand single-buyer Bayesian settings, i.e. the optimal item-pricing, is computationally hard even in single-item settings where the buyer's value distribution is a sum of independently distributed attributes, or multi-item settings where the buyer's values for the items are independent. We also show that it is intractable to optimally price the grand bundle of multiple items for an additive bidder whose values for the items are independent. These difficulties stem from implicit definitions of a value distribution. We provide three instances of how different properties of implicit distributions can lead to intractability: the first is a $\#P$-hardness proof, while the remaining two are reductions from the SQRT-SUM problem of Garey, Graham, and Johnson~\cite{DBLP:conf/stoc/GareyGJ76}. While simple pricing schemes can oftentimes approximate the best scheme in revenue, they can have drastically different underlying structure. We argue therefore that either the specification of the input distribution must be highly restricted in format, or it is necessary for the goal to be mere approximation to the optimal scheme's revenue instead of computing properties of the scheme itself.
\end{abstract}

\section{Introduction}

Designing auctions to maximize revenue in a Bayesian setting is a problem of high importance in both theoretical and applied economics~\cite{milgrom2004putting,myerson1981,Nisan:2007:AGT:1296179}. While substantial progress has been made on designing mechanisms with revenue guarantees that are approximately optimal~\cite{BhattacharyaGGM10,CaiD11,ChawlaHK07,ChawlaHMS10}, the question of determining the optimal mechanism exactly has been much more intricate~\cite{Alaei11,AlaeiFHHM12,CaiDW12,CaiDW12b,DW12,RePEc:huj:dispap:dp542,ManelliV07,optimalsingleitem}.

In this paper, we study the complexity of designing optimal deterministic
auctions for single-bidder problems, i.e. optimal pricing mechanisms. Prior to our work, Briest showed that finding the optimal pricing mechanism for a unit-demand bidder is highly inapproximable when the bidder's values for different items are correlated~\cite{Briest}. Our work complements his by either considering  single-item settings, or multi-item settings with product value distributions. We also investigate the complexity of optimally pricing the grand bundle of multiple items for an additive buyer whose values for the items are independent. For these problems we demonstrate that even when the optimal mechanism can only be one of two possibilities, it can be computationally difficult to determine which one achieves the highest expected revenue. 

\smallskip We note that all hard instances presented in this paper have fully polynomial-time approximation schemes, and thus our results preclude exact algorithms but not computationally efficient approximation schemes. From a practical perspective, a nearly optimal mechanism may be almost as desirable as an exact one. From a theoretical perspective, however, it is important to understand the structure of the exactly optimal mechanism~\cite{myerson1981}, which may be drastically different than that of approximate ones. Computational barriers to determining the best mechanism, such as the ones presented here, reflect barriers to understanding its structure. 


Our results suggest in particular that great care must be taken in how a bidder's value distributions are specified. Intricate distributions can be described succinctly, providing a simple outlet to encode computationally hard problems. We present three concrete scenarios where succinctly-represented distributions lead to computational hardness: Easy-to-describe discrete distributions may have exponential size support, may have mild irrationality in their support, or have mild irrationality in the probabilities they assign. Indeed, many (or all) of these features of discrete distributions can be present in simple continuous distributions. Thus, to obtain a robust theory of optimal Bayesian mechanism design, we must either aim for only approximate revenue guarantees or  severely limit the types and specification format of allowable value distributions.


\section{Preliminaries}

In our model, there is a seller with $n$ items and a buyer whose values for the items $v_1,  ..., v_n$ are random variables drawn from known distributions $F_1,...,F_n$. We will consider both  \emph{unit-demand} and  \emph{additive} buyer types:
\begin{itemize}
\item A (quasi-linear) unit-demand buyer is interested in buying at most one item; if the item prices are $p_1,...,p_n$, the buyer buys the item maximizing his \emph{utility}, $v_i-p_i$, as long as it is positive, breaking ties among the maximizers in some pre-determined way, e.g. lexicographic or in favor of the cheapest/most expensive item.
\item A (quasi-linear) additive buyer values a subset $S$ of items $\sum_{i \in S} v_i$. If subset $S$ is priced $P_S$, his utility for buying that subset is $\sum_{i \in S} v_i - P_S$. The buyer buys the subset of items that maximizes his utility, as long as it is positive, breaking ties among subsets in some pre-determined way.
\end{itemize}

In the case of a unit-demand bidder, the seller's goal is to price the items to optimize the expected price paid by the buyer. Finding the optimal such prices is called the \emph{unit-demand pricing problem}. In the case of an additive bidder, the seller's goal is to price all subsets of items to optimize the expected price paid by the buyer. Of course, the seller may not want to explicitly list the price of every subset but describe their prices in some succinct manner, or may want to offer only some subsets at a finite price. We are particularly interested in the {\em grand bundle pricing problem} where the seller wants to optimally price the set of all items (the grand bundle) and the buyer must take all items or nothing. As shown in \cite{manellivincent}, pricing just the grand bundle is optimal in several natural settings. Furthermore, it oftentimes achieves revenue close to the optimal mechanism~\cite{armstrong,bundling}. Optimally pricing the grand bundle is furthermore interesting in its own right~\cite{fangnorman}.

Finally, the distributions $F_1, ..., F_n$ may be provided explicitly, by listing their support and the probabilities placed on each point in the support, or implicitly giving a closed-form formula for them. In this paper, we study how various ways to describe the distributions affect the complexity of the pricing problem.

\section{Complexity of Sum-of-Attributes Distributions}

We first consider  the problem of optimally pricing a single item for a single buyer whose value for the item is a sum of independent random variables. The probability distribution of the item's value has an exponentially sized support, but has a succinct description in terms of each component variable's distribution. The seller must choose a price $P$ for the item. The buyer will accept the offer (and pay $P$) if his value for it is at least $P$, and will reject the offer (giving the seller zero revenue) if his value is strictly less than $P$. The seller's goal is to choose $P$ to maximize his expected revenue. In fact it follows from Myerson~\cite{myerson1981} that pricing the item at the optimal price is the optimal mechanism in this setting, even among randomized mechanisms.

This problem occurs fairly naturally. When selling a complex product (for example, a car), there are a number of attributes (color, size, etc) that a buyer may or may not value highly, and his value for the product may be the sum of his values for the individual attributes. If his values for the attributes are independent, the buyer's value for the product can be modeled as a sum of independent random variables.

Formally, the problem we study in this section is the following.
\begin{definition}[The Sum-of-Attributes Pricing (SoAP) Problem]
Given $n$ pairs of nonnegative integers $(u_1, v_1), (u_2, v_2), \ldots, (u_n,v_n)$ and rational probabilities $p_1, p_2, \ldots, p_n$, determine the price $P^*$ which maximizes
$P^* \cdot Pr[\sum_{i=1}^n X_i \geq P^*],$
where the $X_i$ are independent random variables taking value $u_i$ with probability $p_i$ and $v_i$ with probability $1-p_i$. 
%
\end{definition}


Notice that we can always view an instance of the sum-of-attributes pricing problem as an instance of the grand bundle pricing problem where we seek the optimal price to sell the ``grand bundle'' of a collection of $n$ items that are independently distributed.


\begin{theorem}
The {Sum-of-Attributes Pricing} problem and the {Grand Bundle Pricing} problem are $\#P$-hard.
\end{theorem}

\begin{proof}
We show how to use oracle access to the SoAP problem to solve the counting analog of the SUBSET-SUM problem, defined next, which is  $\#P$-complete.\footnote{Indeed, the reduction from SAT to SUBSET-SUM as presented in \cite{Sipser2006} is parsimonious.}

\medskip \noindent {$\#$-SUBSET-SUM}: Given as input a set of positive integers $\{a_1, a_2, \ldots, a_{n}\}$ and a positive integer $T \leq \sum_i a_i$, the goal is to determine the number of subsets of the $a_i$'s which sum to at least $T$. 

\medskip \noindent The idea of our reduction is to design an instance of the SoAP problem with n+1 attributes for which the optimal price is one of two possible prices. A single parameter (in particular, the probability $p_{n+1}$ of the last attribute) determines which of these two prices is optimal. By repeatedly querying a SoAP oracle with varying values of $p_{n+1}$, we can determine the exact threshold value of $p_{n+1}$, which provides sufficient information to deduce the answer to the $\#$-subset sum instance.

\smallskip We proceed to provide the details of our reduction. Given an instance of the $\#$-subset sum problem, we create an instance of SoAP with $n+1$ attributes, where for all $i \in \{1,\ldots,n\}$ we take $u_i = a_i$ and $v_i = 0$, while for the last attribute we take $u_{n+1} = T+1$ and $v_{n+1}= 1$.
Moreover, for all $i \in \{1,\ldots,n\}$, we set
$$p_i \triangleq \frac{1}{2^nn(n+1+\sum_{j=1}^n a_j)^2}.$$
Notice in particular that the first $n$ attributes have the same probability of taking their highest value. Moreover, the probability that all the first $n$ attributes have value 0 is:
$$\left(1 -  \frac{1}{2^nn(n+1+\sum_{j=1}^n a_j)^2}\right)^n > 1 -  \frac{1}{2^n(n+1+\sum_{j=1}^n a_j)^2}$$
i.e. very close to 1. We leave the probability $p_{n+1}$ that the last attribute takes its highest value a free parameter, which we denote by $p$ for convenience.

\medskip Now, suppose that we use price $B$ for the SoAP instance. We claim the following:
\begin{enumerate}
	\item If $B = 1$, the expected revenue is 1.
	\item If $1 < B < T+1$, then the expected revenue is at most
	$$B\left(p + \frac{1-p}{2^n(n+1+\sum_{j=1}^n a_j)^2} \right).$$
	\item If $B = T+1$, then the expected revenue is at least $p(T+1)$.
	\item If $ T+1 < B \leq T+1 + \sum_{j=1}^n a_j$, then the expected revenue is at most
	$$\left(T+1 + \sum_{i=1}^na_i \right)\left(\frac{1}{2^n(n+1+\sum_{j=1}^n a_j)^2} \right) \leq \frac{1+ \sum_{j=1}^n a_i}{2^{n-1}(n+1+\sum_{j=1}^n a_j)^2} < 1.$$
	\item If $B > T + 1+ \sum_{j=1}^n a_j$, then the expected revenue is 0.
\end{enumerate}

The fourth and fifth cases are never optimal, since they are both dominated by using $B = 1$. We claim that the second case is also never optimal. Suppose for the sake of contradiction that some integral price $B$ strictly between 1 and $T+1$ were optimal. Then we would have the following two constraints:
\begin{itemize}
	\item $B\left(p + \frac{1-p}{2^n(n+1+\sum_{j=1}^n a_j)^2} \right) \geq 1$
	\item $B\left(p + \frac{1-p}{2^n(n+1+\sum_{j=1}^n a_j)^2} \right) \geq (T+1)p$.
\end{itemize}
To show a contradiction, define for convenience $$\epsilon \triangleq \frac{1}{2^n(n+1+\sum_{j=1}^n a_j)^2}.$$
We will show that no value of $p$ exists for which both of the above constraints are simultaneously satisfied. From the first constraint, we deduce
$p + \epsilon(1-p) \geq 1/B$
and thus
$$p \geq \frac{1/B - \epsilon}{1 - \epsilon} \geq \frac{1/T - \epsilon}{1 - \epsilon} > 1/T - \epsilon,$$
where for the last inequality we used that $T \leq \sum_{j=1}^n a_j$. Moreover,
$$1/T - \epsilon \geq \frac{1}{\sum_{j=1}^n a_i} - \frac{1}{2^n(n+1+\sum_{j=1}^n a_j)^2} \geq \frac{1}{\sum_{j=1}^n a_j} - \frac{1}{2^n\sum_{j=1}^n a_j} \geq \frac{1}{2\sum_{j=1}^n a_j}.$$
Therefore, the first constraint implies that $p > \frac{1}{2 \sum a_j}.$ From the second constraint, we deduce $B(p + \epsilon(1-p)) \geq (T+1)p$
and thus
$$p \leq \frac{B \epsilon}{T+1 - B + B \epsilon},$$
where we used that $B \le T$ so $T+1 - B + B\epsilon > 1$. We further have
$$p < B\epsilon \leq T \epsilon \leq \sum_{j=1}^n a_j \epsilon = \frac{\sum_{j=1}^n a_j}{2^n(n+1+\sum_{j=1}^n a_j)^2}< \frac{1}{2 \sum_{j=1}^n a_j}.$$
We get a contradiction as both constraints on $p$ cannot be satisfied simultaneously. In summary, we have shown the following:

\begin{center}
\emph{``For any $p$, the optimal price is either 1 or $T+1$.''}
\end{center}

We also note the following monotonicity property. If, for some $p$, the optimal price is $T+1$, then the optimal price is $T+1$ for any $p' > p$.\footnote{This follows from the fact that the expected revenue from selling at $T+1$ will only increase as $p$ increases.} Therefore, there exists a unique $p^*$ for which the expected revenue of selling at price $T+1$ is exactly the same as the expected revenue of selling at price 1. 

Suppose that we knew some  $p^*$ such that the expected revenue of selling at $T+1$ is exactly $1$. Then, if we denote by $V_n$ the total value of the first $n$ attributes, $p^*$ should satisfy: 
$$1 = (T+1)\left(p^* + (1-p^*)P[V_n \geq T] \right);$$
so
$$P[V_n \geq T] = \frac{1/(T+1) - p^*}{1 - p^*}.$$
Therefore, it is simple arithmetic to compute $P[V_n \geq T]$ from $p^*$. We also note that
$$P[V_n \geq T] = \sum_{k=0}^n p_1^k (1-p_1)^{n-k} \cdot S(k,T) = p_1^n \cdot \sum_{k=0}^n \left(\frac{1-p_1}{p_1}\right)^{n-k} \cdot S(k,t),$$
where $S(k,T)$ is the number of size $k$ subsets of the $a_i$'s which sum to at least $T$. By our choice of $p_1$ being sufficiently small, we know that $\frac{1-p_1}{p_1} = \frac{1}{p_1} - 1$ is an integer greater than  $2^n$. Therefore, the $S(k,t)$ are the unique integers in the base-$(\frac{1}{p_1}-1)$ representation of $P[V_n \geq T]/p_1^n$, and can be found efficiently. So given $p^*$ we can compute the total number of subsets of the $a_i$'s that sum up to at least $T$, thereby solving the given instance of $\#$-SUBSET SUM.



It remains to argue that we can compute $p^*$ using oracle access to SoAP. We do binary search on $p$ while maintaining all other parameters of the SoAP instance fixed, as described above. In every step of the binary search, we solve the corresponding SoAP instance, determining if the optimal price is 1 or $T+1$ and respectively increasing or decreasing the value of $p$ for the next step, until we have pinned down $p^*$ exactly. To argue that this takes polynomial time we notice that:
$$p^* = \frac{1/(T+1) - P[V_n \geq T]}{1 -  P[V_n \geq T]}.$$
We also notice that $P[V_n \geq T]$ is a rational number that can be specified with a polynomial number of bits.\footnote{In particular, each number of the form $p_1^i(1-p_1)^{n-i}$ has polynomial bit-length.} So $p^*$ has polynomial accuracy and we need polynomially many calls to SoAP to determine it exactly.\qed
\end{proof}

\section{Complexity of Mildly Irrational Valuations}

Issues of numerical precision may arise when analyzing value distributions which are implicitly described. Even very mild irrationality, such as the support of the distribution containing square roots of integers, can cause the resulting pricing problem to be computationally intricate. In particular, optimization may require deciding between two  mechanisms whose expected utility differs only by an exponentially small amount. In this section, we present an example of how we can reduce a numerical problem whose status even in NP remains unknown to the pricing problem for a unit-demand buyer with mildly irrational valuations.

\begin{definition}[The Square Root Sum Problem]
Given positive integers $\alpha_1 \leq \alpha_2 \leq \cdots \leq \alpha_n$ and $K$, the {SQRT-SUM} problem is to determine whether or not $\sum_{i=1}^n \sqrt{\alpha_i} > K$.
\end{definition}
\noindent While known to be in PSPACE, it remains an important open problem whether the square root sum problem is solvable in NP, let alone whether it is in P.~\cite{DBLP:journals/siamcomp/EtessamiY10,DBLP:conf/stoc/GareyGJ76}

\begin{remark}
Checking whether $\sum_i \sqrt{a_i}=K$ for positive integers $a_i$, $i=1,..,n$,
and $K$ can be done in polynomial time~\cite{DBLP:conf/stoc/GareyGJ76}. So the square root sum problem
draws its computational difficulty from instances where equality between $\sum_i
\sqrt{a_i}$ and $K$ does not hold and we need to decide whether $\sum_i \sqrt{a_i}$ is
$>$ or $<$ than $K$. In the hardness proofs of Theorems 2 and 3 we will implicitly
assume that the given instance of the square root sum problem satisfies $\sum_i
\sqrt{a_i} \neq K$. Given such instance we will construct an unit-demand pricing
instance whose solution answers the question of whether $\sum_i
\sqrt{a_i}$ is $>$ or $<$ than $K$.
\end{remark}

\begin{remark}
The important computational difference between the square root of an integer and the sum of square roots of multiple integers is that the $i$-th bit of the former can be computed in time polynomial in $i$ and the number's description complexity, while the same is not known to be true for the latter.
\end{remark}

\begin{theorem} \label{thm:irrational values}
The unit-demand pricing problem is {SQRT-SUM}-hard when the item values are independent of support two with rational probabilities and each possible item value is the square root of an integer.\footnote{The item values are mildly irrational since the $i$-th bit of the square root of an integer can be computed exactly in time polynomial in $i$ and the description complexity of the integer.}
\end{theorem}

\begin{proof}


We will reduce SQRT-SUM to the pricing problem for a single unit-demand buyer whose values for the items are distributed independently, take one of two possible values with rational probabilities, and each of these possible values is the square root of an integer.

\smallskip Given an input $\alpha_1 \leq \alpha_2 \leq \cdots \leq \alpha_n$ and $K$ to the SQRT-SUM problem, we construct an input to the unit-demand pricing problem with $n+1$ items.
 For $i = 1,\ldots, n$, item $i$ has value $\sqrt{\alpha_i}$ with probability $1/i$, and value $0$ with probability $1 - 1/i$. Finally, item $n+1$ has value $T/2$ with probability $1/2 + \epsilon$ and value $T$ with probability $1/2 - \epsilon$, where:
$$\epsilon \triangleq \frac{K}{4n\max({K, \alpha_n})} \leq \frac{1}{2}; \qquad T \triangleq  \frac { (1/2 + \epsilon) K } { n \epsilon}.$$
\noindent Notice that $T/2 > \frac K {4 n \epsilon}= \max({K, \alpha_n}) \geq \alpha_n \geq \sqrt{\alpha_n}$.

\smallskip We now claim that the optimal expected revenue for the unit-demand pricing instance we defined is the maximum of $T/2$ and $$(1/2-\epsilon)T + \frac{1/2 + \epsilon}{n}\left( \sqrt{\alpha_1} + \cdots + \sqrt{\alpha_n} \right).$$

Indeed, it is clearly possible to achieve revenue $T/2$ by pricing item $n+1$ at $T/2$ and all other items at a price greater than $\sqrt{\alpha_n}$. Since $T/2 > \sqrt{\alpha_n}$, if item $n+1$ is priced less than or equal to $T/2$, the revenue cannot be higher than $T/2$.

Now what if item $n+1$ were priced at a price higher than $T/2$? Suppose, e.g., that we price item $n+1$ at $T$ and all other items $i$ at $\sqrt{\alpha_i}$. Then the expected revenue we would get is\footnote{Suppose that ties are broken in favor of the most expensive item.}
\begin{align}
(1/2 - \epsilon)T + (1/2 + \epsilon) \left( \frac{1}{n}\sqrt{\alpha_n} + \frac{n-1}{n}\cdot \frac{1}{n-1}\sqrt{\alpha_{n-1}} + \cdots +\frac{1}{n}\sqrt{\alpha_1} \right)\label{eq:revenue2}
\end{align}
We claim that this is the best revenue we could possibly achieve if item $n+1$ is priced at a price higher than $T/2$. Indeed, it is easy to see that the maximum of the values of items $1,\ldots,n$ is independent of the value of item $n+1$, it has expectation ${1\over n} \sum_i \sqrt{\alpha_i}$ and, because $T/2 > \sqrt{\alpha_n}$, it is smaller than $T$ with probability $1$. So consider any pricing where the price of item $n+1$ is larger than $T/2$. In the event that the value of item $n+1$ is $T$ (which happens with probability exactly $1/2-\epsilon$) the best revenue that the pricing could possibly get is at most $T$, while in the event that the value of item $n+1$ is $T/2$ (which happens with probability exactly $1/2+\epsilon$) the revenue cannot exceed the maximum of the values of items $1,\ldots,n$ which has expectation ${1\over n} \sum_i \sqrt{\alpha_i}$ even after conditioning on the value of item $n+1$ as it is independent from the value of item $n+1$.


Observe that~\eqref{eq:revenue2} is higher than $T/2$ if and only if
$$\epsilon T < \frac{(1/2  + \epsilon)}{n}\left( \sqrt{\alpha_1} + \cdots + \sqrt{\alpha_n} \right),$$
which occurs precisely when $K <   \sqrt{\alpha_1} + \cdots + \sqrt{\alpha_n}.$\qed
\end{proof}

\section{Complexity of Mildly Irrational Probabilities}

The reduction of the previous section used distributions that were supported on irrational values. A possible critique of this in a discrete setting is that it may be unnatural for an individual to hold irrational values for an item. Contrastingly, it seems more natural to allow for a person's values to be rational but to depend on certain mildly irrational probabilities.

Perhaps the simplest form of an irrational probability is one for which we
can efficiently compute arbitrary bits of its binary expansion
correctly.\footnote{This property is satisfied, for example, by a probability
of the form $\sqrt{r}$, where r is a rational number; but, as remarked in section
4, it is unknown whether it is satisfied by a probability of the form $\sum_i
\sqrt{r_i}$, for rational $r_i$'s.} Notice that using a fair coin to sample
exactly such probability, e.g. $\sqrt{1/3}$, is no more work than sampling
exactly a rational probability, e.g. $1/3$: Imagine an infinite sequence of coin
tosses. We reveal a prefix of that sequence until, viewed as a binary number,
we can certify that the sequence lies above or below the target probability
written in binary; if above, we output $1$, otherwise we output $0$.

We now consider unit-demand pricing instances as in the previous section, except where the values are integral and the probabilities are irrational. As in the previous section, we will give a SQRT-SUM-hardness reduction.


\begin{theorem}
The unit-demand pricing problem is {SQRT-SUM}-hard when the item values are independent of support two, have probabilities for which the $i^{th}$ bit of their binary expansions can be computed in time polynomial in $i$, and each possible item value is integral.
\end{theorem}

\begin{proof}
Let $a_1 \le ... \le a_n$ and $K$ be an instance of the SQRT-SUM problem. Also let $X$ be a large integer with $X > \max\{3 K / n,a_n\}$. We define $a_{n+1} = X^2$ maintaining the monotonicity of the sequence $a_i$ since $X > a_n$.

We reduce the given SQRT-SUM instance to an instance of the unit-demand pricing problem with $n+1$
items. For $i = 1,..., n$, item $i$ has value $i$ with probability $p_i = 1 - \sqrt{a_i / a_{i+1}}$, and value $0$ with probability $\sqrt{a_i / a_{i+1}}$. Finally, item $n+1$ has value $T/2$ with probability $3/4$ and value $T$ with probability $1/4$, where:
$$T \triangleq 3 \left(n - {K\over X}\right).$$
Notice that by the choice of $X > 3 K / n$ we have that $T/2 > n$, the highest
possible value of any other item. Also, since the sequence of $a_i$'s is non-decreasing, all probabilities $p_i$ are well defined.

As in the proof of Theorem~\ref{thm:irrational values}, we can argue that the optimal pricing either prices item $n+1$ at $T/2$ and the other items at infinity (call this  ``Scheme 1''),  or
prices all items at their high value (call this ``Scheme 2''). In the former case the revenue is $T/2$. In the latter case the bidder will choose to buy the largest item he values high, i.e. will choose item $n+1$ if he values it high, otherwise item $n$ if he values it high, and so on.\footnote{As in the proof of Theorem~\ref{thm:irrational values} we assume that ties are broken  in favor of the most expensive item.} Therefore, Scheme 1 beats Scheme 2 if and only if:
$$\frac T 2 > \frac T 4  + \frac 3 4 ( p_n n + p_{n-1} (1-p_n) (n-1) + ... + p_1 \prod_{i=2}^n (1-p_i) ),$$
which becomes, after substituting for the $p_i$'s:
$$ \frac{T}{2} > \frac{T}{4} + \frac{3}{4}\sum_{i=1}^n \left(i \left( \sqrt{\frac{a_{i+1}}{a_{n+1}}} - \sqrt{\frac{a_i}{a_{n+1}}} \right) \right).$$
Simplifying and using the fact that $\sqrt {a_{n+1}} = X$, our condition becomes
$$\frac T 2 > \frac T 4  + \frac 3 4 \left( n - \frac { \sum_{i=1}^n \sqrt{a_i} } X\right).$$
This occurs precisely when:
$$ \sum_{i=1}^n \sqrt{a_i} > X ( n - T/3 ) = K.$$

\noindent Therefore, Scheme 1 is strictly better than Scheme 2 precisely when $\sum_{i=1}^n \sqrt { \alpha_i } > K $, concluding our reduction from the SQRT-SUM problem.\qed
\end{proof}

\section{Future Work}
Studying the complexity of optimal pricing in a Bayesian context is an important question, both theoretically and practically. However,  to have a robust complexity model, great care must be taken in specifying the input distributions. Indeed, as shown in this paper, implicit distributions can easily embed hard problems into the distribution's parameters, and therefore any complexity theoretic model of pricing must take into account the complexity of the distributions themselves, and not just the length of a minimal specification.

A setting that avoids the computational barriers raised in this paper is that of several items, each distributed independently on some finite size support, with all values and probabilities rational and explicitly given. This problem is not yet resolved for either unit-demand or additive bidders. Moreover, while our paper has focused only on discrete distributions, issues of distributional specification are perhaps even more vital if one wishes to model the complexity of pricing with continuous distributions. It is of interest to propose a robust computational framework for studying the pricing problem with continuous distributions.

Finally, our results  apply to computing the optimal deterministic mechanism, which in the case of a single buyer is tantamount to finding an optimal pricing scheme. It is an important open question to determine the complexity of the optimal mechanism design problem when randomized mechanisms are also allowed.

\vspace{-7pt}\bibliographystyle{plain}
\bibliography{mybib}

\begin{thebibliography}{10}

\bibitem{Alaei11}
Saeed Alaei.
\newblock {Bayesian Combinatorial Auctions: Expanding Single Buyer Mechanisms
  to Many Buyers}.
\newblock In {\em the 52nd Annual IEEE Symposium on Foundations of Computer
  Science (FOCS)}, 2011.

\bibitem{AlaeiFHHM12}
Saeed Alaei, Hu~Fu, Nima Haghpanah, Jason Hartline, and Azarakhsh Malekian.
\newblock {Bayesian Optimal Auctions via Multi- to Single-agent Reduction}.
\newblock In {\em the 13th ACM Conference on Electronic Commerce (EC)}, 2012.

\bibitem{armstrong}
Mark Armstrong.
\newblock {Price Discrimination by a Many-Product Firm}.
\newblock {\em Review of Economic Studies}, 66(1):151--68, January 1999.

\bibitem{BhattacharyaGGM10}
Sayan Bhattacharya, Gagan Goel, Sreenivas Gollapudi, and Kamesh Munagala.
\newblock {Budget Constrained Auctions with Heterogeneous Items}.
\newblock In {\em the 42nd ACM Symposium on Theory of Computing (STOC)}, 2010.

\bibitem{Briest}
Patrick Briest.
\newblock {Uniform budgets and the envy-free pricing problem}.
\newblock In {\em the 35th International Colloquium on Automata, Languages and
  Programming (ICALP)}, 2008.

\bibitem{CaiD11}
Yang Cai and Constantinos Daskalakis.
\newblock {Extreme-Value Theorems for Optimal Multidimensional Pricing}.
\newblock In {\em the 52nd Annual IEEE Symposium on Foundations of Computer
  Science (FOCS)}, 2011.

\bibitem{CaiDW12}
Yang Cai, Constantinos Daskalakis, and S.~Matthew Weinberg.
\newblock {An Algorithmic Characterization of Multi-Dimensional Mechanisms}.
\newblock In {\em the 44th Annual ACM Symposium on Theory of Computing (STOC)},
  2012.

\bibitem{CaiDW12b}
Yang Cai, Constantinos Daskalakis, and S.~Matthew Weinberg.
\newblock {Optimal Multi-Dimensional Mechanism Design: Reducing Revenue to
  Welfare Maximization}.
\newblock In {\em the 53rd Annual IEEE Symposium on Foundations of Computer
  Science (FOCS)}, 2012.

\bibitem{ChawlaHK07}
Shuchi Chawla, Jason~D. Hartline, and Robert~D. Kleinberg.
\newblock {Algorithmic Pricing via Virtual Valuations}.
\newblock In {\em the 8th ACM Conference on Electronic Commerce (EC)}, 2007.

\bibitem{ChawlaHMS10}
Shuchi Chawla, Jason~D. Hartline, David~L. Malec, and Balasubramanian Sivan.
\newblock {Multi-Parameter Mechanism Design and Sequential Posted Pricing}.
\newblock In {\em the 42nd ACM Symposium on Theory of Computing (STOC)}, 2010.

\bibitem{DW12}
Constantinos Daskalakis and S.~Matthew Weinberg.
\newblock {Symmetries and Optimal Multi-Dimensional Mechanism Design}.
\newblock In {\em the 13th ACM Conference on Electronic Commerce (EC)}, 2012.

\bibitem{DBLP:journals/siamcomp/EtessamiY10}
Kousha Etessami and Mihalis Yannakakis.
\newblock {On the Complexity of Nash Equilibria and Other Fixed Points}.
\newblock {\em SIAM J. Comput.}, 39(6):2531--2597, 2010.

\bibitem{fangnorman}
Hanming Fang and Peter Norman.
\newblock To bundle or not to bundle.
\newblock {\em RAND Journal of Economics}, 37(4):946--963, December 2006.

\bibitem{DBLP:conf/stoc/GareyGJ76}
M.~R. Garey, Ronald~L. Graham, and David~S. Johnson.
\newblock {Some NP-Complete Geometric Problems}.
\newblock In {\em the 8th Annual ACM Symposium on Theory of Computing (STOC)},
  1976.

\bibitem{bundling}
Sergiu Hart and Noam Nisan.
\newblock Approximate revenue maximization with multiple items.
\newblock In {\em the 13th ACM Conference on Electronic Commerce (EC)}, 2012.

\bibitem{RePEc:huj:dispap:dp542}
Omer Lev.
\newblock {A Two-Dimensional Problem of Revenue Maximization}.
\newblock Discussion Paper Series dp542, The Center for the Study of
  Rationality, Hebrew University, Jerusalem, April 2010.

\bibitem{ManelliV07}
A.~M. Manelli and D.~R. Vincent.
\newblock {Multidimensional Mechanism Design: Revenue Maximization and the
  Multiple-Good Monopoly}.
\newblock {\em Journal of Economic Theory}, 137(1):153--185, 2007.

\bibitem{manellivincent}
Alejandro~M. Manelli and Daniel~R. Vincent.
\newblock Bundling as an optimal selling mechanism for a multiple-good
  monopolist.
\newblock {\em Journal of Economic Theory}, 127(1):1--35, 2006.

\bibitem{milgrom2004putting}
P.~Milgrom.
\newblock {\em {Putting Auction Theory to Work}}.
\newblock Cambridge University Press, 2004.

\bibitem{myerson1981}
Roger~B. Myerson.
\newblock {Optimal Auction Design}.
\newblock {\em Mathematics of Operations Research}, 6(1):58--73, 1981.

\bibitem{Nisan:2007:AGT:1296179}
Noam Nisan, Tim Roughgarden, Eva Tardos, and Vijay~V. Vazirani.
\newblock {\em {Algorithmic Game Theory}}.
\newblock Cambridge University Press, New York, NY, USA, 2007.

\bibitem{optimalsingleitem}
Christos~H. Papadimitriou and George Pierrakos.
\newblock {On optimal single-item auctions}.
\newblock In {\em the 43rd annual ACM symposium on Theory of computing (STOC)},
  2011.

\bibitem{Sipser2006}
Michael Sipser.
\newblock {\em {Introduction to the theory of computation: second edition}}.
\newblock PWS Pub., Boston, 2 edition, 2006.

\end{thebibliography}

\end{document}